\documentclass[11pt,english]{article}

\usepackage{latexsym}
\usepackage{epsfig}
\usepackage{amssymb}
\usepackage{array}
\usepackage{float}
\usepackage{color}
\usepackage{dsfont}
\usepackage{verbatim}
\usepackage{enumerate}
\usepackage[font={it}]{caption}
\graphicspath{{figures/}{./}}

\newenvironment{proof}{{\bf Proof. } }{{\hfill $\Box$}\vspace{.5pc}}
\newtheorem{theorem}{Theorem}[section]
\newtheorem{corollary}[theorem]{Corollary}
\newtheorem{definition}[theorem]{Definition}
\newtheorem{lemma}[theorem]{Lemma}

\newtheorem{remark}[theorem]{Remark}

\newtheorem{example}[theorem]{Example}
\floatstyle{ruled}

\floatstyle{ruled}
\newfloat{algo}{thp}{lop}
\floatname{algo}{Algorithm}

\newcommand{\mR}{\mathcal{R}}

\newcommand{\tO}{\mathtt{O}}
\newcommand{\BEGLIST}{\begin{list}{}{\partopsep -3pt \parsep -2pt}}
\newcommand{\ENDLIST}{\end{list}}

\newcommand{\ie}{\emph{i.e., }}
\newcommand{\eg}{\emph{e.g., }}

\begin{document}

\title{Deterministic Leader Election Among Disoriented Anonymous Sensors}
\author{
  Yoann Dieudonn\'e$^\dagger$, Florence Lev\'e$^\dagger$, Franck Petit$^\ddagger$, and Vincent Villain$^\dagger$ \medskip \\
  $^\dagger$ MIS Laboratory, University of Picardie Jules Verne, France\\
  $^\ddagger$ LIP6/Regal, Universit\'e Pierre et Marie Curie, INRIA, CNRS, France
}
\date{}
\maketitle

\begin{abstract}
We address the Leader Election (LE) problem in networks of anonymous sensors
sharing no kind of common coordinate system. Leader Election is a fundamental symmetry breaking problem in distributed computing. Its goal is to assign value 1 (leader) to one of the entities and value 0 (non-leader) to all others.

In this paper, assuming $n>1$ disoriented anonymous sensors, we
provide a complete characterization on the sensors positions to
deterministically elect a leader, provided that all the sensors'
positions are known by every sensor.

More precisely, our contribution is twofold: First, assuming $n$ anonymous
sensors agreeing on a common \emph{handedness} (\emph{chirality}) of their own coordinate system, we provide a complete characterization on the sensors positions to deterministically elect a leader.
Second, we also provide such a complete chararacterization for sensors devoided of a common handedness.

Both characterizations rely on a particular object from combinatorics on words, namely the \emph{Lyndon Words}.
\bigskip

\textbf{Keywords}: Distributed Leader Election, Sense of Direction, Chirality, Sensor Networks, Lyndon Words.

\end{abstract}

%
%
%
%
%
%
%
%


\section{Introduction}

In distributed settings, many problems that are hard to solve become easier 
to solve with a \emph{leader} to coordinate the system.  
The problem of electing a leader among a set of computing units is then one of the 
fundamental tasks in distributed systems.  
The \emph{Leader Election} (LE) Problem consists in distinguishing among a set of entities exactly one of them. 
The leader election problem is covered in depth in many books related to distributed 
systems, \eg \cite{L96,S07}. 

The distributed systems considered in this paper are \emph{sensor networks}.  Sensor networks
are dense wireless networks that are used to collect (to sense) environmental data such as 
temperature, sound, vibration, pressure, motion, etc.  The data are either simply sent toward some data collectors 
or used as an input to perform some basic cooperative tasks.  
Wireless Sensor Networks (WSN) are emerging distributed systems providing diverse services to numerous 
applications in industries, manufacturing, security, environment and habitat monitoring, healthcare, traffic control, etc.  
WSN aim for being composed of a large quantity of sensors as small, inexpensive, and low-powered as possible.  
Thus, the interest has shifted towards the design of distributed protocols for very weak sensors, \ie 
sensors requiring very limited capabilities, \eg \emph{uniformity} (or, \emph{homogeneity} --- all the sensors
follow the same program ---, \emph{anonymity} --- the sensors are \emph{a priori} indistinguishable ---,
\emph{disorientation} --- the sensors share no kind of coordinate system nor common sense of direction. 

However, in weak distributed environments, many tasks have no solution.  In particular, the \emph{Pattern Formation} problem for sensors having the additional capability of \emph{mobility} is not always solvable. 
The Pattern Formation problem consists in the design of protocols allowing autonomous mobile sensors to form 
a specific class of patterns, \eg \cite{SY99,FPSW99,FPSW01,DK02,K05,DP07a,DieudonneLP08}. Such mobile sensors are often 
referred to as \emph{robots} or \emph{agents}. 
In~\cite{FPSW99}, the authors discuss whether the pattern formation problem can be solved or not
according to the capabilities the robots are supposed to have.  
They consider the ability to agree on the direction and orientation of one
axis of their coordinate system (North) (Sense of Direction) and a common \emph{handedness} (\emph{Chirality}). 
Assuming sense of direction, chirality, and \emph{Unlimited Visibility} --- each robot is able to locate all
the robots ---, they show that the robots can form any arbitrary pattern.  Then, they show that
with the lack of chirality, the problem can be solved in general with an odd number of robots only.  With 
the lack of both sense of direction and chirality, the pattern formation problem is unsolvable in general.

In~\cite{FPSW01} and \cite{DieudonnePV10}, the authors show the fundamental relationship between the Pattern Formation problem and the Leader 
Election problem. In the first paper, it is shown that, for a group of $n\geq3$ robots, pattern formation problem cannot be solved if the robots cannot elect a leader.
In the second one, the authors aim at knowing whether the reverse is true or not. In other words, if the robots can elect a leader, can they solve the pattern formation problem? The authors prove that this property holds for $n\geq4$ robots in the asynchronous model CORDA, provided the robots share the same chirality.

In~\cite{FPSW01}, they show that under sense of direction and chirality, 
Leader Election can be solved by constructing a total order over the coordinates of all the agents.  
By contrast, with no sense of direction and lack of chirality, Leader Election is unsolvable (in general).
Informally, the results in~\cite{FPSW99,FPSW01} come from the fact 
that starting from a totally symmetric configuration where the positions of the robots coincide with the vertice of a regular $n$-gon, no robot can be distinguished. Nevertheless, no characterization of the geometric configurations allowing to elect a leader, among disoriented sensors, is given. 

In a first approach, the possibility of electing a leader among disoriented sensors seems to be related to
the absence of some kind of symmetry in the configuration. Besides, this way would be in line with the result of Angluin~\cite{A80} showing that, in uniform interconnected networks, the impossibility of breaking a possible symmetry in the initial configuration makes the leader election unsolvable deterministically. However, this approach does not turn out to be appropriate in our context. Indeed, as depicted in Figure~\ref{fig:po}, some symmetric configurations do not prevent the sensors from distinguishing a leader. This raises the following question : ``\emph{Given a set of weak sensors scattered on the plane, what are the (minimal) geometric conditions to be able to deterministically agree on a single sensor?}''

\begin{figure}[H]
\begin{center}
  \begin{minipage}[t]{0.4\linewidth}
    \centering
    \epsfig{file=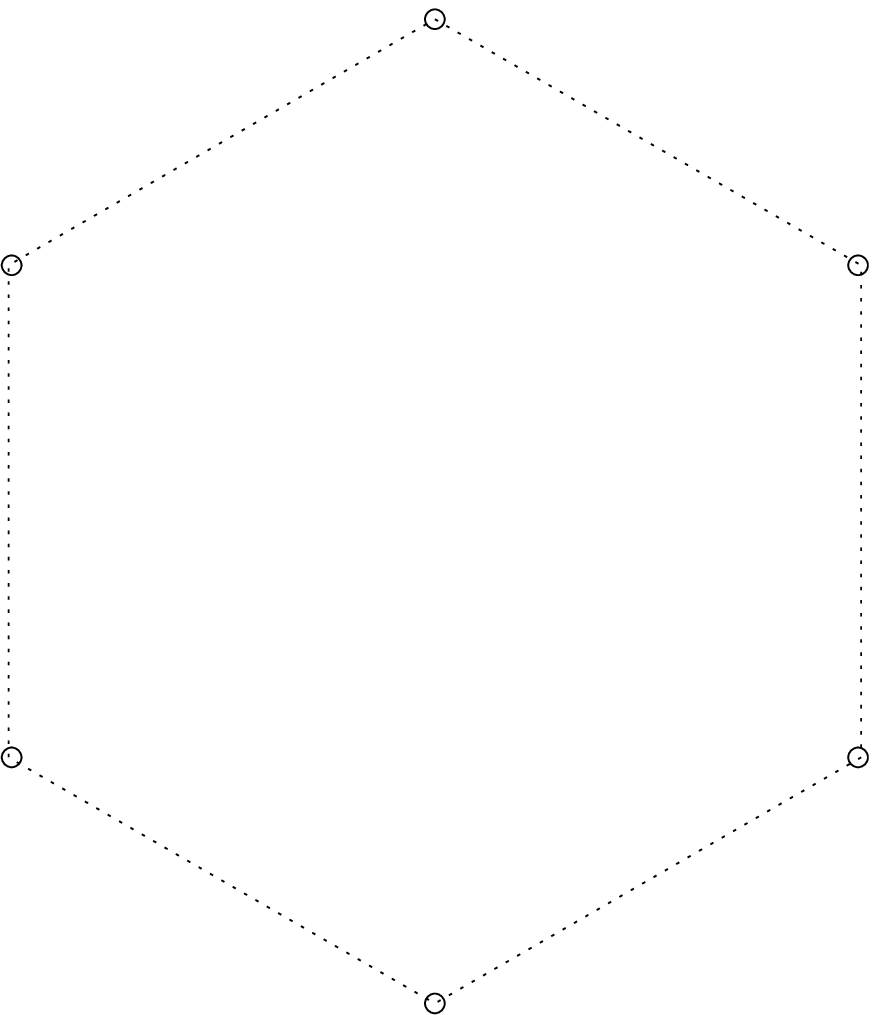, width=0.8\linewidth}\\
    {\footnotesize ($a$) Symmetric Configuration : a regular $6$-gon}
  \end{minipage}
  \begin{minipage}[t]{0.4\linewidth}
    \centering
    \epsfig{file=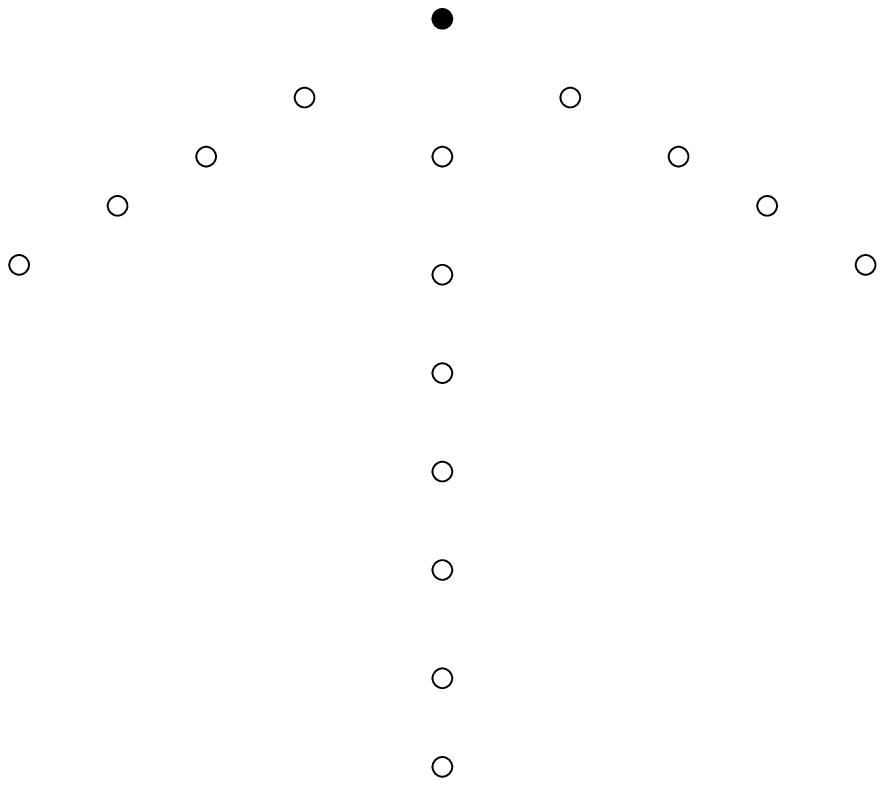, width=0.8\linewidth}\\
    {\footnotesize ($b$) Symmetric Configuration : an arrow}
  \end{minipage}
\end{center}
 \caption{Two examples of symmetric configuration. In Figure~\ref{fig:po}.$a$ no sensors can be distinguished. By contrast, in figure~\ref{fig:po}.$b$ every sensor on the vertical line can be distinguished: For example, the sensors can choose the one located at the top of the arrowhead (black bullet) to play the role of a leader.}
\label{fig:po}
\end{figure}

In this paper, this question is addressed under very weak assumptions: in particular, the sensors 
share no kind of common coordinate system.  More precisely, they are not required to share any 
unit measure, common orientation or direction and even any common sense of rotation, say chirality. 

We provide a complete characterization (necessary and sufficient conditions) 
on the sensors positions to deterministically elect a leader. 
Our result holds for any $n>1$, provided that the sensors know all the positions and are able to make real computations.  
The sufficient condition is shown by providing a deterministic algorithm electing a leader.

The proof is based on the ability for the sensors to construct a \emph{Lyndon word} from the sensors' positions
as an input.  A Lyndon word is a non-empty word strictly smaller in the lexicographic order
than any of its suffixes, except itself and the empty word.  
Lyndon words have been widely studied in the combinatorics of words area~\cite{L83}.  
However, only a few papers consider Lyndon words addressing issues in other 
areas than word algebra, \eg \cite{C04,DR04,SM90,DP07a}.  In~\cite{DP07a}, we 
already showed the power of Lyndon words to build an efficient and simple deterministic protocol to form a 
regular $n$-gon with a prime number $n$ of robots.

In the next section (Section~\ref{sec:prel}), we formally describe the distributed model and 
the words considered in this paper.  
In Section~\ref{sec:chi}, for convenience, we present a preliminary result assuming sensors with chirality. The lack of chirality is
addressed in Section~\ref{sec:nochi}.
Finally, we conclude this paper in Section~\ref{sec:conclu}.

\section{Preliminaries}
\label{sec:prel}

In this section, we define the distributed system considered in this paper.  Next, we review some formal 
definitions and basic results on words and Lyndon words.

\subsection{Model}
\label{sub:model}
Consider a set of $n>1$ \emph{sensors} arbitrarily scattered on the plane such that  
no two sensors are located at the same position.  
A \emph{configuration} ${\mathcal{C}}$ is a set of positions $p_1, \ldots, p_n$ occupied by the sensors.  
The sensors are \emph{uniform} and \emph{anonymous}, i.e, they all execute 
the same program using no local parameter (such as an identity) 
allowing to differentiate any of them.  
However, we assume that each sensor is a computational unit having the ability 
to determine the positions of the $n$ sensors within an infinite decimal precision.  
We assume no kind of communication medium.
Each sensor has its own local $x$-$y$ Cartesian coordinate system defined by two coordinate
axes ($x$ and $y$), together with their \emph{orientations}, 
identified as the positive and negative sides of the axes.  

In this paper, we assume that the sensors have no \emph{Sense of Direction} and we discuss the influence of \emph{Chirality}
in a sensor network.

\begin{definition}[Sense of Direction]
\label{def:sod}
A set of $n$ sensors has sense of direction if the $n$ sensors agree on a common
direction of one axis ($x$ or $y$) and its orientation.  
The sense of direction is said to be \emph{partial}
if the agreement relates to the direction only ---\ie 
they are not required to agree on the orientation. 
\end{definition}

In Figure~\ref{fig:ex}, the sensors have sense of direction in the cases~($a$) and~($b$), 
whereas they have no sense of direction in the cases~($c$) and~($d$).  

Given an $x$-$y$ Cartesian coordinate system, the \emph{handedness} is the way in which 
the orientation of the $y$ axis (respectively, the $x$ axis) is inferred according to 
the orientation of the $x$ axis (resp., the $y$ axis).  

\begin{definition}[Chirality]
A set of $n$ sensors has chirality (or the $n$ sensors have chirality) if the $n$ sensors share the same handedness.
\end{definition}

In Figure~\ref{fig:ex}, the sensors have chirality in the cases~($a$) and~($c$), 
whereas they have no chirality in the cases~($b$) and~($d$).  

\begin{figure}[H]
\begin{center}
  \begin{minipage}[t]{0.4\linewidth}
    \centering
\includegraphics{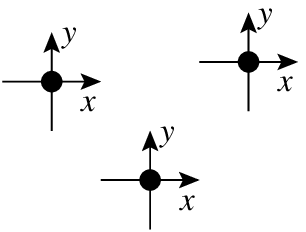}\\
    {\footnotesize ($a$) Sense of Direction and Chirality}
  \end{minipage}
  \begin{minipage}[t]{0.4\linewidth}
    \centering
\includegraphics{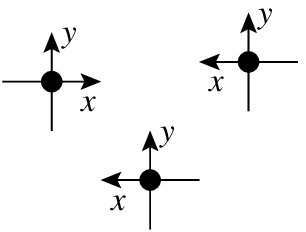}\\
    {\footnotesize ($b$) Sense of Direction and No Chirality}
  \end{minipage}
  \bigskip
\\
  \begin{minipage}[t]{0.4\linewidth}
    \centering
\includegraphics{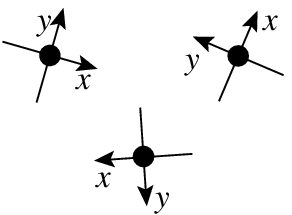}\\
    {\footnotesize ($c$) No Sense of Direction and Chirality}
  \end{minipage}
  \begin{minipage}[t]{0.4\linewidth}
    \centering
\includegraphics{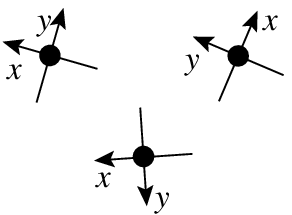}\\ 
  {\footnotesize ($d$) No Sense of Direction and No Chirality}
  \end{minipage}
\end{center}
 \caption{Four examples showing the relationship between Sense of Direction and Chirality}
\label{fig:ex}
\end{figure}

\subsection{Words and Lyndon Words}
\label{sub:LW}

Let an alphabet $A$ be a finite set of letters. 
A non empty \emph{word} $w$ over $A$ is 
a finite sequence of letters $a_1,\ldots,a_i,\ldots,a_\ell$, $\ell>0$.  
The \emph{concatenation} of two words $u$ and $v$, denoted simply by $uv$, is equal to the word
$a_1,\ldots,a_i,\ldots, a_k,b_1, \ldots, b_j,\ldots,b_\ell$ such that $u=a_1,\ldots,a_i,\ldots, a_k$ and 
$v=b_1, \ldots, b_j,\ldots,b_\ell$.  
Let $\varepsilon$ be the \emph{empty word} such that for every word $w$, $w\varepsilon = \varepsilon w = w$.  
The \emph{length} of a word $w$, denoted by $|w|$, is equal to the number of letters of $w$ ($|\varepsilon|=0$).  

Suppose the alphabet $A$ is totally ordered by the relation $\prec$.
Then a word $u$ is said to be \emph{lexicographically} smaller than or equal to a word $v$, 
denoted by $u \preceq v$, iff there exists either a word $w$ such that $v=uw$ or 
three words $r,s,t$ and two letters $a,b$ 
such that $u=ras$, $v=rbt$, and $a\prec b$. We write $u \prec v$ when $u \preceq v$ and $u \neq v$.

Let $k$ and $j$ be two positive integers.  The {\it $k^{th}$ power} of a word $s$ is 
the word denoted by $s^k$ such that $s^0 = \varepsilon$, and $s^k = s^{k-1} s$.
A word $u$ is said to be \emph{primitive} if and only if $u = v^k \Rightarrow k=1$.
Otherwise ($u = v^k$ and $k>1$), $u$ is said to be \emph{strictly periodic}. 


\medskip
The \emph{$j^{th}$ rotation} of a word $w$, notation $Rot_j (w)$, is defined by:

$$ 
Rot_j (w) \stackrel{\mathrm{def}}{=}
\left\{
  \begin{array}{ll}
     \varepsilon & \mbox{if } w = \varepsilon \\
     a_j,\ldots, a_\ell,a_1,\ldots,a_{j-1} & \mbox{otherwise } ( w = a_1,\ldots,a_\ell,\ \ell\geq 1 )
  \end{array}
\right.
$$
    
Note that $Rot_1 (w) = w$. 

\begin{lemma}{\rm \cite{L83}}
\label{lemma:MOT}
Let $w$ and $Rot_j(w)$ be a word and a rotation of $w$, respectively.  
The word $w$ is primitive if and only if  $Rot_j (w)$ is primitive.
\end{lemma}

A word $w$ is said to be {\it minimal} if and only if $\forall j \in 1,\ldots, \ell$,
$w \preceq Rot_j (w)$.

\begin{definition}[Lyndon Word]
\label{def:LW}
A word $w$ $(|w|>0)$ is a Lyndon word if and only if $w$ is nonempty, primitive and minimal, i.e.,
$w \neq \varepsilon$ and $\forall j \in 2,\ldots, |w|,\ w \prec  Rot_{j}(w)$.
\end{definition}

For instance, if $A=\{a,b\}$, then $a$, $b$, $ab$, $aab$, $abb$ are Lyndon words, whereas $aba$, and $abab$ are not ---
$aba$ is not minimal ($aab \preceq aba$) and $abab$ is not primitive ($abab = (ab)^2$). 

\subsection{Leader Election.}
\label{sub:LE}

The \emph{leader election} 
problem considered in this paper is stated as follows:
Given a configuration ${\mathcal{C}}$ of $n>1$ sensors, the $n$ sensors are able to deterministically
agree on a same sensor $L$ called the leader.

\section{Leader Election with Chirality}
\label{sec:chi}

In this section, we consider a sensor network having the property of chirality. The results we propose in this section have been presented in \cite{DieudonneP07}.

In Subsection~\ref{sub:confs&words}, we provide general definitions and establish some results with respect to configurations
and words. We then show in Subsection~\ref{sub:CHI&LE&LW} the relationship between the presence (the lack) of a
leader and the existence (the absence) of a Lyndon word.

\subsection{Configurations and Words}
\label{sub:confs&words}

Given a configuration ${\mathcal{C}}$, $SEC$ denotes the \emph{smallest enclosing circle} of the positions of the sensors in ${\mathcal{C}}$. The center
of $SEC$ is denoted by $\tO$.  The distance from $\tO$ to any point on $SEC$, the radius of $SEC$, is denoted
by $\sigma$. 
In any configuration ${\mathcal{C}}$, $SEC$ is unique and can be computed by any sensor in linear time~\cite{M83,W91}. 
It passes either through two of the positions that are on the same diameter (opposite positions), or through at least 
three of the positions in ${\mathcal{C}}$. 
Note that if $n=2$, then $SEC$ passes both sensors and no sensor can be located inside $SEC$, in particular at $\tO$.

Given a smallest enclosing circle $SEC$, the radii are the line segments 
from the center $\tO$ of $SEC$ to the boundary of $SEC$. 
Let $\mR$ be the finite set of radii such that a radius $r$ belongs to $\mR$ iff at least 
one sensor is located on $r$ but $\tO$.  
Denote by $\sharp\mR$ the number of radii in $\mR$.
In the sequel, we will abuse language by considering radii in $\mR$ only. 
Given two distinct positions $p_1$ and $p_2$ located on the same radius $r$ ($\in \mR$),
$d(p_1,p_2)$ denotes the Euclidean distance between $p_1$ and $p_2$.  

Since $\mR$ and the set of positions are finite, the set of different distances between positions of sensors on the radii in $\mR$ is finite. Every robot codes each of these distances $x$ using a fonction $Code(x)$ such that the output of $Code(x)$ is a letter of an arbitrary alphabet provided with an order relation, and the code respects the natural order on distances (i.e. for letters $a_1$, $a_2$ and distances $d_1$, $d_2$, $a_1 \leq a_2$ if and only if $d_1 \leq d_2$). A similar code is used for the angles in Definition~\ref{def:cw}.

\begin{definition}[Radius Word] 
\label{def:radword}
Let $p_0$ be the position of $\tO$, $p_1, \ldots, p_k$ be the respective positions of 
$k$ sensors $(k \geq 1)$ located on the same radius $r \in \mR$.
Let $\rho_r$ be the radius word such that 
$$ 
\rho_r \stackrel{\mathrm{def}}{=} 
\left\{
  \begin{array}{ll}
     0 & \mbox{ if there exists one sensor at } \tO \\
     a_1 a_2\ldots a_k\mbox{ s.t. }
        \forall i \in [1,k], a_i = Code(\frac{d(p_{i-1},p_i)}{\sigma}) & \mbox{ otherwise}
  \end{array}
\right.
$$

\end{definition}

Note that all the distances are computed by each sensor 
with respect to its own coordinate system, \ie proportionally to its own measure unit.
So, for any radius $r \in \mR$, all the sensors compute the same word $\rho_r$.
For instance, in Figure~\ref{fig:SC2}, for every sensor, $\rho_{r_1}=c$ and $\rho_{r_2} = \rho_{r_3} = ab$. 

\begin{remark}
\label{rem:radword}
If there exists one sensor on $\tO$ ($n > 2$), then for every radius $r \in \mR$, $\rho_r = 0$.
\end{remark}

Denote by $\circlearrowright$ an arbitrary orientation of $SEC$, \ie $\circlearrowright$ denotes either 
the clockwise or the counterclockwise direction\footnote{Note that the clockwise (resp. counterclockwise) direction may be the counterclockwise (resp. clockwise) direction for the robots depending on their handedness. However, without loss of generality, we can assume that both directions are those we commonly use.}.  
Given an orientation $\circlearrowright$ of $SEC$, $\circlearrowleft$ denotes the opposite orientation.
Let $r$ be a radius in $\mR$,
the successor of $r$ in the $\circlearrowright$ direction, denoted by $Succ(r,\circlearrowright)$, 
is the next radius in $\mR$, according to $\circlearrowright$. The $i^{th}$ successor of $r$, denoted by 
$Succ_i(r,\circlearrowright)$, is the radius such that $Succ_0(r,\circlearrowright)=r$, 
and $Succ_i(r,\circlearrowright)=Succ(Succ_{i-1}(r,\circlearrowright),\circlearrowright)$.
Given $r$ and its successor $r'= Succ(r,\circlearrowright)$,  
$\sphericalangle(r \tO r')^\circlearrowright$ denotes the angle between $r$ and $r'$ following the $\circlearrowright$ direction.

\begin{definition}[Configuration Word]
\label{def:cw}
Let $r$ be a radius in $\mR$.  
The configuration word of $r$ with respect to $\circlearrowright$, an arbitrary orientation of $SEC$, 
denoted by $\omega(r)^\circlearrowright$, is defined by:
$$
\omega(r)^\circlearrowright \stackrel{\mathrm{def}}{=}
\left\{
\begin{array}{ll}
(0,0) & \mbox{if } \rho_r = 0\\
(\rho_0,\alpha_0)(\rho_1,\alpha_1)\ldots(\rho_k,\alpha_k) \ \mbox{ with } k=\sharp\mR -1 &\mbox{otherwise}\\
\end{array}
\right.
$$
such that, $\rho_0$ is the radius word of $r$, $\alpha_0$ is a code (similar as for the distances) of the angle between $r$ and its successor in $\mR$ 
following $\circlearrowright$, and for every $i \in [1,k]$, 
$r_i$ being the $i^{th}$ successor of $r$ in $\mR$ following $\circlearrowright$,  
$\rho_i$ is the radius word of $r_i$ and $\alpha_i$ is the code of the angle between $r_i$ and its successor 
in $\mR$ following $\circlearrowright$.
\end{definition}

In the sequel, by abuse of language, we will refer to distances and angles without explicitely mentionning the code, to ease the reading.

\medskip
In Figure~\ref{fig:SC2}, assuming that $\circlearrowright$ denotes the clockwise direction, then:\newline
\noindent
$\omega(r_1)^\circlearrowright = (c,\beta)(ab,\alpha)(ab,\beta)$,
$\omega(r_2)^\circlearrowright = (ab,\alpha)(ab,\beta)(c,\beta)$, 
$\omega(r_3)^\circlearrowright = (ab,\beta)(c,\beta)(ab,\alpha)$,\newline
\noindent  
$\omega(r_1)^\circlearrowleft  = \omega(r_1)^\circlearrowright $,
$\omega(r_2)^\circlearrowleft  = \omega(r_3)^\circlearrowright$, and
$\omega(r_3)^\circlearrowleft  = \omega(r_2)^\circlearrowright$.

\begin{figure}[H]
\begin{center}
\includegraphics[width=0.35\linewidth]{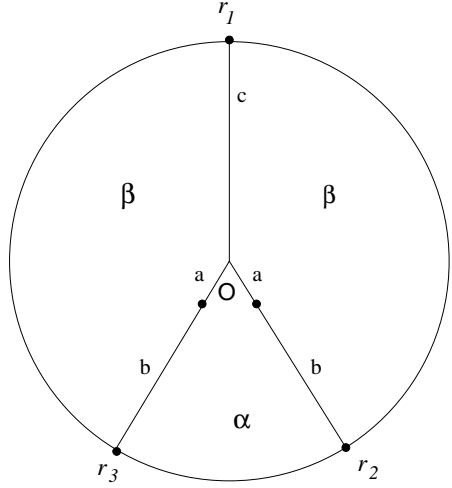}
\end{center}
\caption{Computation of Configuration Words --- the sensors are the black bullets.} 
\label{fig:SC2}
\end{figure}

\begin{definition}
A  configuration $\mathcal{P}=\{p_1, \ldots, p_n\}$ of sensors is said to be {\it periodic} if there exists a rotation $Rot$ of center $\tO$ and of angle $\alpha$, $0 < \alpha < 2\pi$ such that $Rot(\mathcal{P})=\mathcal{P}$, (i.e, the configuration $Rot(\mathcal{P})$ is superimposed on $\mathcal{P}$). Let ${\mathcal{C}}$ be a configuration with no robot at the position $\tO$. A construction of configuration words on ${\mathcal{C}}$ is said to be {\it period-free} when the configuration words corresponding to ${\mathcal{C}}$ are not periodic when ${\mathcal{C}}$ is not periodic. In other words, if the configuration words corresponding to a configuration  ${\mathcal{C}}$ are strictly periodic, then  ${\mathcal{C}}$ is periodic.
\end{definition}

In Figure~\ref{fig:proof}, configuration words are strictly periodic (we can distinct two configuration words, $w=((d,\alpha), (ab,\alpha))^2$ and $w'=((ab,\alpha)ad,\alpha))^2$) and the figure is periodic since the set of positions of the sensors is preserved by the rotation of center $\tO$ and of angle $\pi$.

Notice the following important fact, which is fondamental for the statement of the following results:

\begin{remark}\label{rmk:rotations}
The construction of configuration words stated in Definition~\ref{def:cw} is period-free.
\end{remark}

\begin{proof}
Let $\mathcal{P}=\{p_1, \ldots, p_n\}$ be a configuration of sensors.
If $w(r)$ is strictly periodic for some radius $r$, that is $w(r)=u^k$ for some integer $k>2$ and a word $u$, then there exists $k$ rotations $Rot_j$, $1\leq j \leq k$ of center $\tO$ and of angle $\frac{2\pi j}{k}$ such that $Rot_j(\mathcal{P})=\mathcal{P}$.
\end{proof}

In this section, we consider sensors with chirality.  In particular, they are able to agree on a common 
orientation $\circlearrowright$ of $SEC$.
Let $CW$ be the set of configuration words computed by the sensors over $\mR$
according to $\circlearrowright$.   
Since for every sensor, there is no ambiguity on $\circlearrowright$, in the sequel of this section, 
we omit $\circlearrowright$ in the notations when it is clear in the context.
For instance in Figure~\ref{fig:SC2}, 
$$
\begin{array}{ccccc}
  CW = \{ &
    (c,\beta)(ab,\alpha)(ab,\beta), & 
    (ab,\alpha)(ab,\beta)(c,\beta), &
    (ab,\beta)(c,\beta)(ab,\alpha) &
    \} \\
          &
    = \omega(r_1) &
    = \omega(r_2) &
    = \omega(r_3) 
\end{array}
$$

\begin{remark}
\label{rem:cw}
The following propositions are equivalent:
\begin{enumerate}
\item There exists one sensor on $\tO$
\item For every radius $r \in \mR$, $\omega(r)^\circlearrowright  = \omega(r)^\circlearrowleft = (0,0)$
\item $CW = \{(0,0)\}$
\end{enumerate}
\end{remark}

\begin{remark}
\label{rem:rot}
Given an orientation $\circlearrowright$, for every pair $u,v$ in $CW$, $v$ is a rotation of $u$.   
\end{remark}

The lexicographic order $\preceq$ on the set of radius words over $\mR$ is naturally built over the natural order $<$ on the set of codes of distances (we recall that is is compatible with the natural order on real numbers).

\begin{definition} \label{def:order}
Let $Alph(CW)$ be the set of letters appearing in $CW$.  
Let $(u,x)$ and $(v,y)$ be any two letters in $Alph(CW)$.  
Define the order $\lessdot$ over $Alph(CW)$ as follows:
$$ 
(u,x) \lessdot (v,y) \Longleftrightarrow \left\{
  \begin{array}{ll}
     u \precneqq v\\
     \mbox{or}\\
     u = v \mbox{ and } x < y
     \end{array}
\right.
$$
The lexicographic order $\preceq$ on $CW$ is built over $\lessdot$.
\end{definition}

\begin{lemma}
\label{lemma2}
Given an orientation $\circlearrowright$, if there exist two distinct radii $r_1$ and $r_2$ in $\mR$ such that 
both $\omega(r_1)$ and $\omega(r_2)$ are Lyndon words ($\omega(r_1),\omega(r_2) \in CW$), then $CW = \{(0,0)\}$. 
\end{lemma}

\noindent
\begin{proof}
Assume by contradiction that, given an orientation $\circlearrowright$, two distinct radii $r_1$ and $r_2$ exist 
such that both $\omega(r_1)$ and $\omega(r_2)$ are Lyndon words and $CW \neq \{(0,0)\}$.  
By Remark~\ref{rem:cw}, there exists no sensor located at $\tO$.
By Remark~\ref{rem:rot}, $\omega(r_1)$ (respectively, $\omega(r_2)$) is a rotation of $\omega(r_2)$ (resp.
$\omega(r_1)$).  So, by Definition~\ref{def:cw}, $\omega(r_1) \prec \omega(r_2)$ and $\omega(r_2) \prec \omega(r_1)$. A contradiction.
\end{proof}

\begin{corollary}
\label{cor:lem2}
Given an orientation $\circlearrowright$, if there exists (at least) one radius $r$ such that 
$w(r)$ is a Lyndon word strictly greater than $(0,0)$, then $r$ is unique.
\end{corollary}

\subsection{Chirality, Leader Election and Lyndon Words}
\label{sub:CHI&LE&LW}

In this section, we show that the problems of the existence of a unique Lyndon word and of the existence of a deterministic Leader Election protocol are equivalent.

\begin{lemma}
\label{lemma3}
If there exists $r \in \mR$ such that $w(r)$ is a Lyndon word, then the $n$ sensors are able to deterministically 
agree on a same sensor $L$. 
\end{lemma}
\begin{proof}
Directly follows from Corollary~\ref{cor:lem2}: If there is a sensor $s$ located on $\tO$, then 
the $n$ sensors are able to agree on $L=s$.  Otherwise, there exists a single radius $r \in \mR$ such 
that $\omega(r)$ is a Lyndon word.  In that case, all the sensors are able to agree on the sensor on $r$ that 
is the nearest one from $\tO$. 
\end{proof}

\begin{lemma}
\label{lemma4}
If there exists no radius $r \in \mR$ such that $\omega(r)$ is a Lyndon word, then 
there exists no deterministic algorithm allowing the $n$ sensors to 
agree on a same sensor $L$. 
\end{lemma}
\begin{proof}
Assume by contradiction that no radius $r \in \mR$ exists such that $\omega(r)$ is a Lyndon word and that
there exists an algorithm $A$ allowing the $n$ sensors to deterministically agree on a same sensor $L$. From
Remark~\ref{rem:cw}, there is no sensor located on $\tO$, otherwise for all $r \in \mR$, $w(r)=(0,0)$ would
be a Lyndon word.

Let $min_\omega$ be a word in $CW$ such that $\forall r \in \mR$, $min_\omega \preceq \omega(r)$. That is,
$min_\omega$ is minimal. 
Assume first that $min_\omega$ is primitive.  Then, $min_\omega$ is a Lyndon word that contradicts the assumption. 
So, $min_\omega$ is a strictly periodic word (there exist $u$ and $k>1$ such that $min_\omega=u^k$) and, 
from Lemma~\ref{lemma:MOT}, we deduce that for all $r \in \mR$, 
$w(r)$ is also strictly periodic. Since by Remark~\ref{rmk:rotations} our construction is period-free, 
this implies that the configuration is periodic. Thus, for every $r \in \mR$, there exists at least one 
radius $r' \in \mR$ such that $r \neq r'$ and $\omega(r) = \omega (r')$. So if an algorithm elects a leader on
$r$ in a deterministical way, this algorithm distinguishes another leader on $r'$. 
In that case, $A$ cannot allow the $n$ sensors to deterministically agree on a same sensor $L$. 
\end{proof}

Figure~\ref{fig:proof} presents an example of a configuration where the sensors have the same measure unit and
their $y$ axis meets the radius on which they are located. In this case, no radius $r \in \mR$ exists such
that $\omega(r)$ is a Lyndon word and, with respect to Lemma~\ref{lemma4}, sensors are not able to
deterministically agree on a same sensor $L$. 

\begin{figure}[H]
\begin{center}
\includegraphics[width=0.4\linewidth]{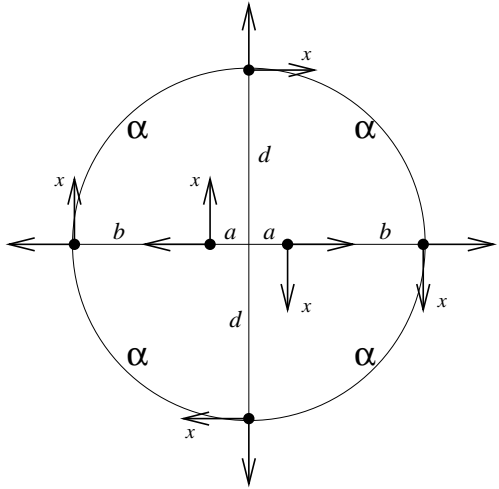}
\end{center}
\caption{An example of configuration as in Lemma~\ref{lemma4}.}
\label{fig:proof}
\end{figure}

\medskip

The following theorem follows from Lemmas~\ref{lemma3} and~\ref{lemma4}:

\begin{theorem}
\label{theorem1}
Given a configuration ${\mathcal{C}}$ of any number $n > 1$ of disoriented sensors with chirality scattered on the plane, 
the $n$ sensors are able to deterministically agree on a same sensor $L$ if and only 
if there exists a radius $r \in \mR$ such that $\omega(r)$ is a Lyndon Word.
\end{theorem}

The acute reader will have noticed that when $n=1$ electing a leader
is straightforward. On the other hand, when $n=2$, it is impossible to
elect a unique leader in a deterministic way: Our theorem confirms
this fact as the two words obtained for the two radii are the same in such a configuration (actually each of both words consists in the concatenation of two identical letters).

\section{Leader Election without Chirality}
\label{sec:nochi}
 
In this section, the lack of chirality among the sensors is considered.
A first result presented in \cite{DieudonneP07} shows that the Configuration Word construction described in Section~\ref{sec:chi} is not enough to deal with an even number of sensors without chirality. We now show how to extend the Configuration Word so that we obtain a general characterization.

In Subsection~\ref{sub:words&nochi}, general definitions and results are given.  
They are then used to show in Subsection~\ref{sub:NOCHI&LE&LW} the relationship between Leader Election and
Lyndon words.

\subsection{Words and Lack of Chirality}
\label{sub:words&nochi}

Devoid of chirality, the sensors are no longer able to agree on a common orientation $\circlearrowright$ of $SEC$.   
In other words, with respect to their handedness, some of the $n$ sensors may assign $\circlearrowright$
to the actual clockwise direction, whereas some other may assign
$\circlearrowright$ to the counterclockwise direction.  
As a consequence, given a radius $r \in \mR$ and two sensors $s$ and $s'$ on $r$, the configuration word
$\omega(r)^\circlearrowright$ computed by $s$ may differ from the one computed by $s'$, depending on whether
$s$ and $s'$ share the same handedness or not.  For instance, in Figure~\ref{fig:SC2}, 
the configuration word $\omega(r_2)^\circlearrowright$ (as defined by Definition~\ref{def:cw}) is equal to either 
$(ab,\alpha)(ab,\beta)(c,\beta)$ or $(ab,\beta)(c,\beta)(ab,\alpha)$, depending on the fact
$\circlearrowright$ denotes the clockwise direction or the counterclockwise, respectively. 

Furthermore, given a radius $r_i\in \mR$, if there exists a radius $r_j\in \mR$ such that 
$\omega(r_i)^\circlearrowleft = \omega(r_j)^\circlearrowright$, then there exists an axial 
symmetry which corresponds to the bissectrix of $\sphericalangle(r_i \tO r_j)$.
Otherwise (\ie no radius $r_j$ exists such that $\omega(r_i)^\circlearrowleft = \omega(r_j)^\circlearrowright$), 
no such axial symmetry exists.  
This observation leads us to provide the following definition in order to classify the radii.

\begin{definition}[Type of symmetry]
\label{def:type}
A radius $r_i\in \mR$ is of type (of symmetry)~$1$ if there exists $j \neq i$ such that $\omega(r_i)^\circlearrowleft = \omega(r_j)^\circlearrowright$.
Otherwise (when  $\omega(r_i)^\circlearrowleft = \omega(r_j)^\circlearrowright \Rightarrow i=j$, that is there is no symmetry axis or the axis corresponds to the radius itself), $r_i$ is said to be of type~$0$.
A radius of type~$t$ is said to be $t$-symmetric.   
\end{definition}

We use the type of symmetry to define the strong configuration word of any radius in $\mR$ as follows:

\begin{definition}[Strong Configuration Word]
\label{def:CW}
Let $r$ be a radius in $\mR$.  
The strong configuration word of $r$ with respect to $\circlearrowright$, an arbitrary orientation of $SEC$, 
denoted by $W(r)^\circlearrowright$, is defined by:
$$
W(r)^\circlearrowright \stackrel{\mathrm{def}}{=}
\left\{
\begin{array}{ll}
(0,0,0) & \mbox{if } \rho_r = 0\\
(s_0,\rho_0,\alpha_0)(s_1,\rho_1,\alpha_1)\ldots(s_k,\rho_k,\alpha_k) \ \mbox{ with } k=\sharp\mR -1 &\mbox{otherwise}\\
\end{array}
\right.
$$
such that, $s_0$ is the type of $r$, $\rho_0$ is the radius word of $r$, $\alpha_0$ is the angle between $r$ and its successor in $\mR$ 
following $\circlearrowright$, and for every $i \in [1,k]$, 
$r_i$ being the $i^{th}$ successor of $r$ in $\mR$ following $\circlearrowright$,  
$s_i$ is the type of symmetry of $r_i$, 
$\rho_i$ is the radius word of $r_i$, and $\alpha_i$ is the angle between $r_i$ and its successor 
in $\mR$ following $\circlearrowright$.
\end{definition}

Note that the (weak) configuration words used in Section~\ref{sec:chi} (Definition~\ref{def:cw}) could have
been defined as strong configuration words by using Constant $0$ as type of symmetry, \ie by mapping each
letter $(\rho,\alpha)$ to $(0,\rho,\alpha)$.  We discuss this issue in more details in the next
section.

Notice also that the construction of strong configuration words has obviously the same important property as the construction of weak configuration words:

\begin{remark}\label{rmk:rotations_nochi}
The construction of strong configuration words stated in Definition~\ref{def:CW} is period-free.
\end{remark}

\begin{definition}[Type of word]
A word $W(r)$ is said to be of type $t$ if and only if its first letter is equal to $(t,x,y)$ for any $x,y$.
\end{definition}

Devoid of chirality, the sensors cannot agree on a common orientation of $SEC$.  So, for each radius $r \in \mR$, 
each sensor computes two strong configuration words, one for each direction, 
$W(r)^\circlearrowright$ and $W(r)^\circlearrowleft$. 
Let $CW$ be the set of strong configuration words computed in both directions by the sensors for each radius in $\mR$.  
In Figure~\ref{fig:SC2}, 
all the sensors compute the following set: 
$$
\begin{array}{ccccc}
  CW = \{ &
    (0,c,\beta)(1,ab,\alpha)(1,ab,\beta), & 
    (1,ab,\alpha)(1,ab,\beta)(0,c,\beta), &
    (1,ab,\beta)(0,c,\beta)(1,ab,\alpha) &
    \} \\
          &
    = W(r_1)^\circlearrowright = W(r_1)^\circlearrowleft &
    = W(r_2)^\circlearrowright = W(r_3)^\circlearrowleft &
    = W(r_3)^\circlearrowright = W(r_2)^\circlearrowleft 
\end{array}
$$

\begin{remark}
\label{rem:CW}
The following propositions are equivalent:
\begin{enumerate}
\item There exists one sensor on $\tO$
\item For every radius $r \in \mR$, $W(r)^\circlearrowright  = W(r)^\circlearrowleft = (0,0,0)$
\item $CW = \{(0,0,0)\}$
\end{enumerate}
\end{remark}


Similarly to Section~\ref{sec:chi}, we build the lexicographic order $\preceq$ on the set of radius words over
$\mR$ as follows:

\begin{definition} \label{def:order2}
Let $Alph(CW)$ be the set of letters appearing in $CW$.  
Let $(b,u,x)$ and $(c,v,y)$ be any two letters in $Alph(CW)$.  
Define the order $\lessdot$ over $Alph(CW)$ as follows:
$$ 
(b,u,x) \lessdot (c,v,y) \Longleftrightarrow \left\{
  \begin{array}{ll}
     b < c\\
     \mbox{or}\\
     b=c \mbox{ and } u \precneqq v\\
     \mbox{or}\\
     b=c \mbox{ and } u = v \mbox{ and } x < y
     \end{array}
\right.
$$
The lexicographic order $\preceq$ on $CW$ is built over $\lessdot$.  
\end{definition}


The proof of the following lemma is similar to the one of Lemma~\ref{lemma2}:

\begin{lemma}
\label{lemma2-bis}
Given an orientation $\circlearrowright$, if there exist two distinct radii $r_1$ and $r_2$ in $\mR$ such that 
both $W(r_1)$ and $W(r_2)$ are Lyndon words ($W(r_1),W(r_2) \in CW$), then $CW = \{(0,0,0)\}$. 
\end{lemma}


\begin{corollary}
\label{cor:lem2-bis}
Given an orientation $\circlearrowright$, if there exists (at least) one radius $r$ such that 
$W(r)$ is a Lyndon word strictly greater than $(0,0,0)$, then $r$ is unique.
\end{corollary}

\subsection{Leader Election and Lyndon Words}
\label{sub:NOCHI&LE&LW}

Let $\mR_L$ be the subset of radii $r \in \mR$ such that $W(r)$ is a Lyndon word in the clockwise or in 
the counterclockwise orientation. Denote by $\sharp \mR_L$ the number of radii in $\mR_L$. 

The following result directly follows from Lemma~\ref{lemma2-bis} and Corollary~\ref{cor:lem2-bis}:
\begin{lemma}
\label{lemma8}
If for some orientation $\circ$ of $SEC$ in $\{\circlearrowright,\circlearrowleft\}$, there exists $r \in \mR$
such that  $W(r)^\circ=(0,0,0)$, then for all $r \in \mR$,  $W(r)^\circ=(0,0,0)$ and the leader is the sensor
at the center of $SEC$.
\end{lemma}

Note that Lemma~\ref{lemma8} applies necessarily when $\sharp \mR_L > 2$. Indeed, this means that there exists at least two distinct radii $r_1$ and $r_2$ such that $w(r_1)$ and $w(r_2)$ are Lyndon words for a common orientation $\circ \in \{\circlearrowright,\circlearrowleft\}$ and by Lemma~\ref{lemma2}, this implies that $CW^\circ=\{(0,0)\}$, so that $CW=\{(0,0)\}$. So, in the remainder of this section, we study the case $\sharp \mR_L \leq 2$.

\begin{lemma}
If for all $r \in \mR$, $W(r)^\circlearrowright \neq (0,0,0)$,  $W(r)^\circlearrowleft \neq (0,0,0)$, and
$R_L=\{r_\ell\}$, the $n$ sensors are able to deterministically agree on a same sensor $L$.
\end{lemma}

\begin{proof}
The leader is the nearest sensor to $\tO$, on $r_\ell$.
\end{proof}

Note that, in this case, $r_\ell$ is necessarily $0$-symmetric.
%
An example of such a configuration is given by Figure~\ref{fig:SC2}: 
$W^\circlearrowright(r_1) = W^\circlearrowleft(r_1) = (0, c, \beta)(1, ab, \alpha)(1, ab, \beta)$ 
is a Lyndon word.  The leader is the sensor on $r_1$.

\begin{lemma}\label{lem:2Lyndon}
If for all $r \in \mR$, $W(r)^\circlearrowright \neq (0,0,0)$,  $W(r)^\circlearrowleft \neq (0,0,0)$, 
and $R_L=\{r_1, r_2\}$, the $n$ sensors are able to deterministically agree on a same sensor $L$ if 
and only if $r_1$ and $r_2$ are of type $0$.
\end{lemma}

\begin{proof}
Without loss of generality we can consider that $W(r_1)^\circlearrowright$ and $W(r_2)^\circlearrowleft$ are Lyndon words. We have three cases to consider:

\begin{enumerate}[1.]

\item
Radii $r_1$ and $r_2$ cannot be of different types. Assume by
contradiction they can. Without loss of generality we can consider
that $r_1$ is of type $0$ and $r_2$ is of type $1$. So, from Definition~\ref{def:type}, we deduce there is another radius of type $1$ which is also a Lyndon word. So, we would have three radii corresponding to a Lyndon word, which contradicts Corollary~\ref{cor:lem2-bis}.

\item
For the same reason, if $r_1$ and $r_2$ are both $1$-symmetric, then necessarily 
$W(r_1)^\circlearrowright = W(r_2)^\circlearrowleft$. In this case, if an algorithm 
elects a leader on $r_1$ in a deterministical way, this algorithm distinguishes another 
leader on $r_2$---refer to Figure~\ref{fig:CE} for an example of such a configuration. 
So an algorithm cannot allow the $n$ sensors to deterministically agree on a same sensor $L$.

\item
If $r_1$ and $r_2$ are both $0$-symmetric, then  $W(r_1)^\circlearrowright \neq W(r_2)^\circlearrowleft$. Without loss of generality, 
we can consider that $W(r_1)^\circlearrowright < W(r_2)^\circlearrowleft$. Then the $n$ sensors are able to deterministically agree 
on the sensor that is the nearest to $\tO$ on $r_1$.
\end{enumerate}
\end{proof}

\begin{example}

Figure~\ref{fig:CE} shows a configuration where each radius is $1$-symmetric; and so no leader exists. Notice that a case like this one can occur only when the number of sensors is even.

\begin{figure}[H]
\begin{center}
\includegraphics[width=0.35\linewidth]{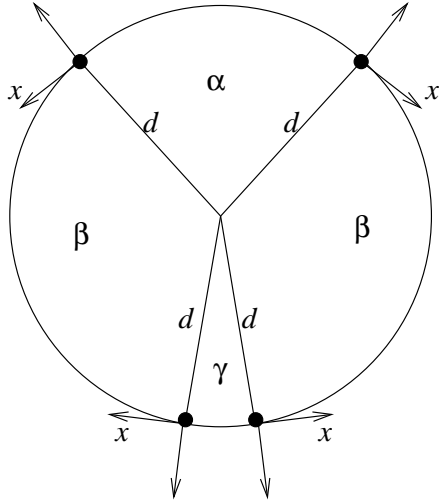}
\end{center}
\caption{An example of configuration where no leader exists.}
\label{fig:CE}
\end{figure}
\end{example}

\begin{example}

Figure~\ref{fig:ex5} shows a configuration as in the last case of the proof of Lemma~\ref{lem:2Lyndon}. We have:
$$W^\circlearrowright(r_1)=(0,c,\alpha)(0, ab, \gamma)(0,ab,\beta) ,\, W^\circlearrowleft(r_1)=(0,c,\beta)(0, ab, \gamma)(0,ab,\alpha)$$
$$W^\circlearrowright(r_2)=(0,ab,\gamma)(0, ab, \beta)(0,c,\alpha) ,\, W^\circlearrowleft(r_2)=(0,ab,\alpha)(0, c, \beta)(0,ab,\gamma)$$
$$W^\circlearrowright(r_3)=(0,ab,\beta)(0, c, \alpha)(0,ab,\gamma) ,\, W^\circlearrowleft(r_3)=(0,ab,\gamma)(0, ab, \alpha)(0,c,\beta)$$

$W^\circlearrowleft(r_2)$ and $W^\circlearrowright(r_3)$ are Lyndon word, $W^\circlearrowleft(r_2)$ is the smallest one, so the leader is the sensor that is the nearest to $\tO$ on $r_2$.

\begin{figure}[!htbp]
\begin{center}
\includegraphics[width=0.35\linewidth]{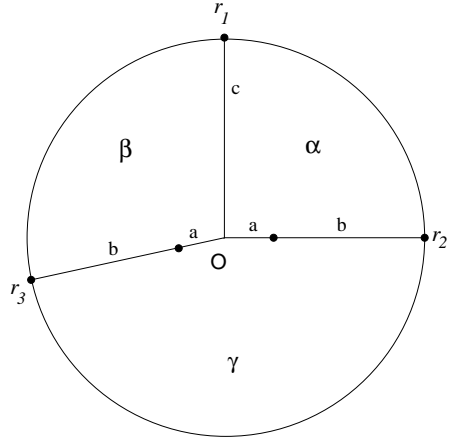}
\end{center}
\caption{The leader is the nearest sensor to $\tO$ on $r_2$.}
\label{fig:ex5}
\end{figure}

\end{example}

\begin{lemma}
If for all $r \in \mR$, $W(r)^\circlearrowright \neq (0,0,0)$,  $W(r)^\circlearrowleft \neq (0,0,0)$, and $R_L=\emptyset$, then there exists no deterministic algorithm allowing the $n$ sensors to agree on a same sensor $L$.
\end{lemma}

\begin{proof}
Assume by contradiction that there exists no radius $r \in \mR$ on the center of $SEC$ or such that $W^\circlearrowleft(r)$ or $W^\circlearrowright(r)$ is a Lyndon word and that there exists an algorithm $A$ allowing the $n$ sensors to deterministically agree on a same sensor $L$.
Let $min^\circ_W$ be a word in $CW$ such that $\forall r \in \mR$, $min^\circ_W \preceq W(r)$. Suppose w.l.o.g. that  $\circ=\circlearrowleft$. That is, $min^\circlearrowleft_W$ is minimal. Assume first that $min^\circlearrowleft_W$ is primitive. Then it is a Lyndon word that contradicts the assumption.
So, $min^\circlearrowleft_W$ is a strictly periodic word and from Lemma~\ref{lemma:MOT}, we deduce that for all $r \in \mR$, $w^\circlearrowleft(r)$ is also strictly periodic. Since by Remark~\ref{rmk:rotations_nochi} our construction is period-free, this implies that the configuration is periodic. Thus, for every $r \in \mR$, there exists at least one radius $r' \in \mR$ such that $r \neq r'$ and $W^\circlearrowleft(r)=W^\circlearrowleft(r')$.

So if an algorithm elects a leader on $r$ in a deterministical way, this algorithm distinguishes another leader on $r'$. In that case, $A$ cannot allow the $n$ sensors to deterministically agree on a same sensor $L$. 
\end{proof}
  

The four previous lemmas lead to:

\begin{theorem}\label{theorem2}
Given a configuration ${\mathcal{C}}$ of any number $n \geq 2$ of disoriented sensors without chirality scattered on the plane, the
$n$ sensors are able to deterministically agree on a same sensor $L$ if and only if there exists $W(r) \in
CW$ such that $r$ is of type 0 and $W(r)$ is a Lyndon word of type $0$.
\end{theorem}

\begin{remark}
Notice that our construction does not necessarily give the same leader for one given configuration whether there is chirality or not. For instance in Figure~\ref{fig:SC2}, with chirality, if the orientation is $\circlearrowleft$, the chosen leader is the nearest sensor to $\tO$ on $r_3$ and if the orientation is $\circlearrowright$, it is the nearest sensor to $\tO$ on $r_2$. Without chirality, the chosen leader is the sensor on $r_1$.
\end{remark}

As noticed in this section, the (weak) configuration words built in
Section~\ref{sec:chi} (Definition~\ref{def:cw}) can be easily defined in terms of strong configuration words
(Definition~\ref{def:CW}) by mapping each pair $(\rho,\alpha)$ to $(0,\rho,\alpha)$. 
Thus, Theorem~\ref{theorem1} also holds over the strong configuration words built on $CW$ assuming chirality. 

So, we can formulate Theorems~\ref{theorem1} and~\ref{theorem2} more generally:

\begin{theorem}
Given a configuration $C$ of any number $n \geq 2$ of disoriented sensors scattered on the plane and a period-free construction of strong configuration words $w$, the $n$ sensors are able to deterministically agree on a same sensor $L$ if and only if there exists a Lyndon word of type $0$ among the configuration words.
\end{theorem}

\section{Conclusion}
\label{sec:conclu}

In this paper, we studied the leader election problem in networks of anonymous sensors sharing no kind of
common coordinate system. Because of the anonymity, the problem is impossible to solve in a deterministic way,
in general. Nevertheless, some specific geometric configurations allows to elect a leader. Using properties of Lyndon words, we give a complete characterization on the sensors positions to deterministically 
elect a leader for any number $n>1$ of anonymous sensors, assuming chirality or not.
This result is based on a specific coding of planar configurations
involving Lyndon words which ensures to keep the properties allowing to distinguish a unique leader, if any, while not creating new ones.
In the light of our study, the leader provided by our construction turns out to be one of the sensors closest to the center of $SEC$ among either all the sensors or only those that are located on the unique symmetry axis of a configuration. In the case where the sensors share a common handedness, this feature has been used to facilitate the resolution of the arbitrary pattern formation problem as it is shown in \cite{DieudonnePV10}. We conjecture that it can also be used in the case where the agents do not share a common handedness.
However it should be noted that, according to the distributed tasks to be solved, it might be sometimes more judicious to obtain a leader having other geometrical characteristics. This is particularly the case for the flocking tasks among mobile agents (which consist in moving as a group while following a leader) which a leader initially closest to the boundary of $SEC$ would be more adapted to achieve.
Therefore, it is worth emphasizing that we may slightly modify our coding in order to take into account the following distributed task to be achieved: For example by giving higher priority to angles instead of distance or more generally, by changing the order over the set of letters.

As a future work, we will concentrate on finding similar characterizations for other collaborative tasks in mobile sensor networks such as localization problem for which we know that no solution exists in general \cite{DieudonneLP10}. We will also study how the solvability of the election problem impacts the solvability of other fundamental distributed tasks in sensor networks.
\bibliographystyle{alpha} 
\bibliography{ngon}

\end{document}